\documentclass[11pt]{article}
\usepackage[utf8]{inputenc}
\usepackage{amsmath,amsfonts,amsthm,amssymb,url}
\usepackage{graphicx}
\usepackage{color}
\usepackage{fullpage}
\usepackage{thmtools, thm-restate}
\usepackage{comment}
\usepackage{breqn}
\usepackage{braket}

\usepackage[english]{babel}
\usepackage[dvipsnames]{xcolor}
\usepackage{hyperref}
\hypersetup{
    colorlinks = true,
    linkbordercolor = {white},
    citecolor = {cyan}
}
\newtheorem{theorem}{Theorem}

\newtheorem{proposition}{Proposition}
\newtheorem{lemma}{Lemma}

\newtheorem{corollary}{Corollary}
\newtheorem{definition}{Definition}[section]

\usepackage[capitalize,nameinlink,noabbrev]{cleveref}
\crefformat{equation}{#2(#1)#3}
\crefmultiformat{equation}{(#2#1#3)}%
{ and~#2(#1)#3)}{, #2(#1)#3}{ and~#2(#1)#3}
\crefname{assumption}{Assumption}{Assumptions}
\crefname{problem}{Problem}{Problems}
\crefname{claim}{Claim}{Claims}
\crefname{fact}{Fact}{Facts}

\newcommand{\F}{\mathbb{F}}
\newcommand{\bv}[1]{\mathbf{#1}}
\newcommand{\norm}[1]{\|#1\|}

\newcommand{\tr}{\mathrm{tr}}

\newcommand{\poly}{\mathrm{poly}}

\DeclareMathOperator{\Z}{\mathbb{Z}}

\title{A Note on Fine-Grained Quantum Reductions for Linear Algebraic Problems}
\date{}
\author{
Kyle Doney \\ UMass Amherst \\ \texttt{kylepdoney@gmail.com} \and 
Cameron Musco  \\ UMass Amherst \\ \texttt{cmusco@umass.edu} 
}

  \usepackage{nth}
  \usepackage{intcalc}

  \newcommand{\arXiv}[1]{\url{http://arxiv.org/abs/#1}}

\begin{document}

\maketitle

\begin{abstract}
We observe that any $T(n)$ time algorithm (quantum or classical) for several central linear algebraic problems, such as computing $\det(\bv A)$, $\tr(\bv A^3)$, or $\tr(\bv A^{-1})$ for an $n \times n$ integer matrix $\bv A$, yields a $O(T(n)) + \tilde O(n^2)$ time \emph{quantum algorithm} for $n \times n$ matrix-matrix multiplication. That is, on quantum computers, the complexity of these problems is essentially equivalent to that of matrix multiplication. Our results follow by first observing that the Bernstein-Vazirani algorithm gives a direct quantum reduction from matrix multiplication to computing $\tr(\bv A\bv B\bv C)$ for $n \times n$ inputs $\bv A,\bv B,\bv C$. We can then reduce $\tr(\bv A\bv B\bv C)$ to each of our problems of interest.

For the above problems, and many others in linear algebra, their fastest known algorithms require $\Theta(n^\omega)$ time, where $\omega \approx 2.37$ is the current exponent of fast matrix multiplication. Our finding shows that any improvements beyond this barrier would lead to faster quantum algorithms for matrix multiplication. Our results complement existing reductions from matrix multiplication in algebraic circuits \cite{burgisser2013algebraic}, and reductions that work for standard classical algorithms, but are not tight -- i.e., which roughly show that an $O(n^{3-\delta})$ time algorithm for the problem yields an $O(n^{3-\delta/3})$ matrix multiplication algorithm \cite{williams2010subcubic}.
\end{abstract}

\section{Introduction}\label{section:intro}

For many linear algebraic problems, such as matrix inversion, linear system solving, determinant computation, and the computation of function traces like $\tr(\bv{A}^3)$ or $\tr(\bv A^{-1})$, the best known algorithms for general $n \times n$ input matrices run in $O(n^\omega)$ time, where $\omega \approx 2.37$ is the exponent of fast matrix multiplication \cite{alman2024refined}. That $O(n^\omega)$ time is an upper bound for these problems is not surprising: all of them can be readily reduced to matrix multiplication with small overhead \cite{burgisser2013algebraic,demmel2007fast}. But the question is -- why can't we do better? Why do all of these problems seem to be as hard as matrix multiplication?

\smallskip

\noindent\textbf{Reductions for Matrix-Output Problems.}
For some problems, where the output is an $n \times n$ matrix, one can readily answer this question by proving that their asymptotic complexity is equivalent to matrix multiplication via a reduction. E.g, we can observe that for any $\bv A, \bv B$, $$\small{\begin{pmatrix}
\bv I & \bv A & \bv 0 \\
\bv 0 & \bv I & \bv B \\
\bv 0 & \bv 0 & \bv I
\end{pmatrix}^{-1} = 
\begin{pmatrix}
\bv I & -\bv A & \bv A\bv B \\
\bv 0 & \bv I & -\bv B \\
\bv 0 & \bv 0 & \bv I
\end{pmatrix}}.$$
That is, by forming the $3n \times 3n$ matrix on the lefthand side and inverting it, one can read off the matrix product $\bv A \bv B$. So, any algorithm for matrix inversion running in faster than $O(n^\omega)$ time would directly yield a faster matrix multiplication algorithm. 

\smallskip

\noindent\textbf{Reductions for Scalar-Output Problems.} For problems like $\det(\bv A)$ or $\tr(\bv A^3)$, where the output is a scalar, establishing equivalence to matrix multiplication is harder. In the algebraic circuit model, one can use the Baur-Strassen theorem \cite{baur1983complexity} to show that a circuit of size $T(n)$ for these and many related problems yields an $O(T(n))$ sized matrix multiplication circuit, establishing equivalence \cite{pan1978strassen,burgisser2013algebraic}.
However, our understanding outside algebraic circuits is more limited.

Williams and Williams \cite{williams2010subcubic} show via a reduction that an $O(n^{3-\delta})$ time algorithm for $\tr(\bv A^3)$ implies an $O(n^{3-\delta/3})$ time algorithm for Boolean matrix multiplication.\footnote{The reduction of \cite{williams2010subcubic} applies even to the easier problem of {triangle detection} in an undirected graph, which is equivalent to checking if $\tr(\bv A^3) > 0$ when $\bv A$ is binary and symmetric.} One has $\delta \le 1$ since any algorithm requires  $\Omega(n^2)$ time on general input instances. Thus, their reduction at best yields an $O(n^{2.66})$ time algorithm for Boolean matrix multiplication, which is worse than the state-of-the-art $O(n^\omega)$ time.
Nevertheless, their result gives evidence of the hardness of $\tr(\bv A^3)$: any faster algorithm would yield a completely new subcubic time algorithm for matrix multiplication.

\cite{musco2017spectrum} use the fact that, for symmetric $\bv A$, $\tr(\bv A^3)$ equals the sum of cubed eigenvalues of $\bv A$ to show that the reduction of \cite{williams2010subcubic} also applies to computing (even to small approximation error) many other quantities of interest that also depend on $\bv A$'s eigenvalue spectrum, including $\det(\bv A)$, $\tr(\bv A^{-1})$, $\tr(\exp(\bv A))$, $\tr(\bv A^p)$ for $p \ge 3$, various matrix norms, and beyond.

However, all of these results inherit the looseness of the \cite{williams2010subcubic} reduction. Thus, establishing tight reductions from matrix multiplication to the multitude of scalar-output linear algebra problems whose best known complexity is $O(n^\omega)$ remains open outside the algebraic circuit model.

\subsection{Tight Quantum Reductions From Matrix Multiplication}

We observe that for many linear algebraic problems, tight reductions from matrix multiplication can be given by leveraging quantum computing. We first prove a result for $\tr(\bv A^3)$, when $\bv A$ is symmetric with bounded integer entries.\footnote{Requiring $\bv A$ to be symmetric only makes our bound stronger, as an algorithm for general $\bv A$ must work for any symmetric $\bv A$. Requiring symmetry facilitates our later reductions to other matrix functions (\Cref{cor:intro}).} A similar result holds for matrices over finite fields. Throughout, we use $\poly(x)$ to denote $x^c$ for any fixed constant $c > 0$. We use $\tilde O(\cdot)$ to suppress polylogarithmic factors in the input argument.
\begin{theorem}\label{thm:intro}
    If there exists an algorithm (classical or quantum) that, given symmetric $\bv A \in \Z^{n \times n}$ with entries bounded in magnitude by $\poly(n)$ computes $\tr(\bv A^3)$ in time $O(T(n))$, then there exists a quantum algorithm that, given $\bv A, \bv B \in \Z^{n \times n}$ with entries bounded in magnitude by $\poly(n)$, computes $\bv A \bv B$ in time $O(T(n)) + \tilde O(n^2)$.
\end{theorem}
\Cref{thm:intro} shows that the complexity of matrix multiplication and $\tr(\bv A^3)$ on quantum computers are essentially equivalent. There has been significant work on both fast and query efficient quantum algorithms for Boolean matrix multiplication in particular \cite{williams2010subcubic,le2012improved,le2012time,buhrman2004quantum,jeffery2016improving,le2014quantum}. However, to the best of our knowledge, no quantum algorithm for multiplying general matrices with Boolean, bounded integer, or finite field entries is known with runtime faster than $O(n^\omega)$. \Cref{thm:intro} can thus be viewed as a conditional lower bound: any progress on faster algorithms for $\tr(\bv A^3)$ would lead to progress on fast quantum matrix multiplication.

We can apply the reductions of \cite{musco2017spectrum} to extend \Cref{thm:intro} to  several other problems. These reductions hold even for algorithms that compute the function $f(\bv A)$ to polynomially small relative error -- i.e., that output $\tilde f$ with $\left |\frac{\tilde f - f(\bv A)}{f(\bv A)} \right | \le \frac{1}{\poly(n)}$. Considering approximation algorithms is natural as in some cases, even when $\bv A$ is integer, $f(\bv A)$ may be non-integer and even e.g., irrational. 

\begin{corollary}\label{cor:intro}
If there exists an algorithm (classical or quantum) that, given symmetric $\bv A \in \Z^{n \times n}$ with entries bounded in magnitude by $\poly(n)$ computes a $\frac{1}{\poly(n)}$ relative error approximation to any of the following functions in time $O(T(n))$, then there exists a quantum algorithm that, given $\bv A, \bv B \in \Z^{n \times n}$ with entries bounded in magnitude by $\poly(n)$, computes $\bv A \bv B$ in time $O(T(n)) + \tilde O(n^2)$.
\begin{itemize}
    \item $\det(\bv A)$, $\log\det(\bv A)$
    \item $\tr(\bv A^{-1})$ for non-singular $\bv A$, $\tr(\exp(\bv A))$, $\tr(\bv A^p)$ for integer $p \ge 3$
    \item The $p^{th}$ Schatten norm, $\|\bv A\|_p^p = \sum_{i=1}^n \sigma_i(\bv A)^p$ for any fixed $p \neq 1,2$.
    \item The SVD entropy $\sum_{i=1}^n \sigma_i(\bv A) \cdot \log\sigma_i(\bv A))$ for non-singular $\bv A$.
    \end{itemize}
    Above $\sigma_1(\bv A),\ldots, \sigma_n(\bv A) \ge 0$ are the singular values of $\bv A$.
\end{corollary}

\subsection{Proof Summary}

\Cref{thm:intro} follows from an application of the powerful Bernstein-Vazirani theorem \cite{bernstein1993quantum}, which shows that, given black-box access to a linear function $f(\bv x) = \langle \bv s, \bv x \rangle$, one can recover the coefficient vector $\bv s$ using a single quantum query to the function with an equal superposition of all inputs.

First, we observe that the problem of computing $\tr(\bv A \bv B \bv C)$ given $n \times n$ integer matrices $\bv A$, $\bv B$, and $\bv C$ can be reduced to computing $\tr(\bv{A}^3)$ with just a constant factor overhead. In particular, one can check that $\tr(\bv A \bv B \bv C) = \frac{1}{6} \tr(\bv M^3)$  for the $3n \times 3n$ symmetric matrix 
\begin{equation*}
    \small
\bv M = 
\begin{pmatrix}
0 & \bv A & \bv C^T \\
\bv A^T & 0 & \bv B \\
\bv C & \bv B^T & 0
\end{pmatrix}.
\end{equation*}

With this observation in hand, we turn our attention to reducing matrix multiplication to computing $\tr(\bv A \bv B \bv C)$. We can write $\tr(\bv A\bv B \bv C) = \sum_{i,j} (\bv A \bv B)_{ij} \bv C_{ji}$ as a linear function in the entries of $\bv C$. In the algebraic circuit model, this immediately yields a reduction from matrix multiplication: given a circuit of size $T$ for computing $\tr(\bv A \bv B \bv C)$, by the Baur-Strassen theorem \cite{pan1978strassen,baur1983complexity}, we can construct a circuit of size $O(T)$ that computes all partial derivatives of this function as well. Since $\frac{\partial \tr(\bv A \bv B \bv C)}{\partial \bv C_{ji}} = (\bv A \bv B)_{ij}$, this circuit thus directly outputs all entries of the matrix product $\bv A \bv B$.

For a quantum reduction, we will use exactly the same idea, except instead of applying Baur-Strassen, we apply the Bernstein-Vazirani theorem. That is, we think of $\bv A$ and $\bv B$ as fixed and use a single query to the linear function $f(\bv C) = \tr(\bv A \bv B \bv C)$ (for which we have assumed we have an efficient algorithm) with an equal superposition over all possible $\bv C$. From this query, we can recover all coefficients of the function -- i.e., all entries of $\bv A \bv B$. Accounting for the $\tilde O(n^2)$ cost overhead of implementing this method, including preparing the equal superposition of all inputs, yields \Cref{thm:intro}. Using the techniques of \cite{musco2017spectrum} we can reduce computing $\tr(\bv A^3)$, and in turn $\tr(\bv A \bv B \bv C)$ to many other scalar-output linear algebra problems, yielding \Cref{cor:intro}.




\section{Preliminaries}\label{sec:prelim}

\subsection{Notation}\label{sec:notation}

Let $\F_q$ denote the finite field containing $q$ elements. Let $\Z$ denote the integers. We represent matrices with bold uppercase letters -- e.g., $\bv A$. $\bv A_{ij}$ denotes the $(i,j)$ entry of $\bv A$. 
 We denote the inner product over vectors $\bv x,\bv y$ in either $\F_q^n$ or $\mathbb{Z}^n$ as $\langle \bv x,\bv y \rangle=\sum_{i=1}^n \bv x_i \bv y_i$. 
 Recall that for $n \times n$ matrices $\bv X$ and $\bv Y$ the natural inner product is given by $\langle \bv X, \bv Y \rangle = \tr(\bv X \bv Y^T) = \sum_{i=1}^n \sum_{j=1}^n \bv X_{ij} \bv Y_{ij}$. Plugging in $\bv X = \bv A \bv B$ and $\bv Y = \bv C^T$, we have $\tr(\bv A \bv B \bv C) = \sum_{i=1}^n \sum_{j=1}^n (\bv A \bv B)_{ij} \bv C_{ji}.$ 

\subsection{Computational Model and Representation of Numerical Values}

Our results use the Bernstein-Vazirani algorithm to reduce matrix multiplication in a blackbox way to computing $\tr(\bv A \bv B \bv C)$. All other reductions, from $\tr(\bv A \bv B \bv C)$ to $\tr(\bv A^3)$ and other linear algebraic problems are classical. Thus, we require only that, given an underlying algorithm for computing a function with runtime $T(n)$, we can implement a \emph{phase oracle}  (see \Cref{app:quantum} for a formal definition) for computing that function on a quantum computer with cost $O(T(n))$. If the underlying algorithm is implemented in a standard random-access machine, this means that we need access to a quantum random-access memory with unit-cost read/write operations, along with standard quantum gates that can be run at unit cost. Prior work on quantum algorithms for linear algebraic computation uses similar models \cite{le2012time,chen2025quantum,jeffery2016improving}.

For simplicity, we focus on integer input matrices with entries bounded in magnitude by $n^c$ for any fixed $c > 0$, or entries in $\F_q$ with $q$ bounded by $n^c$. The entries of these matrices and their products can be represented exactly using $O(\log(n))$ bits -- note that if $\bv A, \bv B \in \Z^{n \times n}$ have entries bounded in magnitude by $n^c$, $\bv A\bv B$ has entries bounded in magnitude by $n^{2c+1}$. Arithmetic operations (e.g., multiplication, addition) on these entries can be performed in $\tilde O(\log(n))$ time.

Our result in \Cref{cor:intro} involves functions like $\tr(\bv A^{-1})$, $\log \det(\bv A)$, and $\det(\bv A)$  that do not necessarily take integer values, and/or may take values exponentially large in $n$. Nevertheless, we can represent a $\frac{1}{\poly(n)}$ relative error approximation (as required by \Cref{cor:intro}) to any of these quantities using a floating point representation with an $O(\log n)$ bit significand and an $O(\log n$ bit exponent. We can again perform arithmetic operations on such representations in $\tilde O(\log n)$ time, which is all our reductions require. 

\subsection{The Bernstein-Vazirani Theorem} 

The key quantum primitive that we use is the Bernstein-Vazirani theorem \cite{bernstein1993quantum}, which allows recovering the coefficients of a linear function with just a single quantum query to the function. Traditionally, the theorem is stated for functions over $\F_2^n$. We use a simple generalization to $\F_q^n$, which we prove for completeness in \Cref{app:quantum}. Our results for bounded integer matrices follow by working over a sufficiently large finite field to represent the range of possible integer outputs.

\begin{theorem}[Bernstein-Vazirani Theorem]\label{thm:ExtendedBV}
For any $f:\mathbb{F}_q^n\rightarrow \mathbb{F}_q$ of the form $f(\bv x)=\langle \bv s, \bv x \rangle$, where $\bv s\in\mathbb{F}_q^n$ is a fixed vector, there exists a quantum algorithm that outputs $\bv s$ with probability $1$ using just a single call to an algorithm computing $f$ and $\tilde O(n \log q)$ additional time.
\end{theorem}

As a direct consequence of \Cref{thm:ExtendedBV}, we have the following reduction from matrix multiplication over $\F_q$ to computing the trilinear form $\tr(\bv A \bv B \bv C)$, 
which can be viewed as the function $f(\bv C) = \sum_{i=1}^n \sum_{j=1}^n (\bv A \bv B)_{ij} \bv C_{ji} = \langle \bv A \bv B, \bv C^T\rangle$. This result mirrors what can be shown in the algebraic circuit model using e.g., the Baur-Strassen theorem \cite{pan1978strassen,baur1983complexity}.
\begin{theorem}[Matrix Mult. to Trilinear Form]\label{theorem:QuantTrace}
If there exists an algorithm (classical or quantum) that, given $\bv A, \bv B, \bv C \in \F_q^{n \times n}$ for $q = \poly(n)$ computes $\tr(\bv A\bv B\bv C)$ in $T(n)$ time, then there exists a quantum algorithm that given $\bv A,\bv B \in \mathbb{F}_q^{n \times n}$, computes $\bv A\bv B$ in $T(n) + \tilde O(n^2)$ time.
\begin{proof}
    We directly apply \Cref{thm:ExtendedBV} where we `flatten' $\bv A \bv B$ and $\bv C$ and represent them as vectors $\bv s$ and $\bv x$ respectively in $\F_q^{n^2}$. A single call to our algorithm for computing $f(\bv C) = \tr(\bv A \bv B \bv C) = \langle \bv s, \bv x \rangle$ allows us to recover $\bv s \in \F_q^{n^2}$, which we can then rearrange into the matrix product $\bv A \bv B \in \F_q^{n \times n}$. 
\end{proof}

\end{theorem}

\section{Main Results}\label{sec:main}

We now prove \Cref{thm:intro} and \Cref{cor:intro}. We start by stating several straightforward non-quantum reductions in \Cref{sec:classical} -- from $\tr(\bv A \bv B \bv C)$ to $\tr(\bv A^3)$ and in turn to several scalar-output linear algebra problems. We then prove the main theorems in \Cref{sec:mainProofs}

\subsection{Classical Linear Algebraic Reductions}\label{sec:classical}

We first show that the problem of computing $\tr(\bv A \bv B \bv C)$ can be reduced to computing $\tr(\bv M^3)$ for a constant sized larger symmetric matrix $\bv M$.

\begin{proposition}[Trilinear Form to $\tr(\bv A^3)$]\label{prop:trace}
If there exists a $T(n)$ time algorithm that given any symmetric $\bv M \in \F_q^{n \times n}$ for $q = \poly(n)$ outputs $\tr(\bv M^3)$, then there exists a $T(3n) + \tilde O(n^2)$ time algorithm that, given any $\bv A,\bv B,\bv C \in \F_q^{n \times n}$
outputs $\tr(\bv A\bv B\bv C)$.
\end{proposition}

\begin{proof}

Given $\bv A$, $\bv B$, $\bv C \in \F_q^{n \times n}$, we first form the $3n \times 3n$ symmetric matrix: 
\begin{equation*}\small
\bv M = 
\begin{pmatrix}
0 & \bv A & \bv C^T \\
\bv A^T & 0 & \bv B \\
\bv C & \bv B^T & 0
\end{pmatrix}.
\end{equation*}
This takes $\tilde O(n^2)$ time, accounting for the $O(\log n)$ bit complexity of representing our matrix entries. The cube of $\bv M$ is equal to:

\begin{align*}\small
\bv M^3 = 
\begin{pmatrix}
0 & \bv A & \bv C^T \\
\bv A^T & 0 & \bv B \\
\bv C & \bv B^T & 0 \\
\end{pmatrix}^3 &= 
\begin{pmatrix}
\bv{AA}^T + \bv {C}^T \bv C & \bv C^T \bv B^T& \bv A\bv B \\
\bv B\bv C & \bv{A}^T \bv A + \bv {BB}^T & \bv A^T \bv C^T \\
\bv B^T \bv A^T & \bv C\bv A & \bv{CC}^T + \bv{B}^T\bv B \\
\end{pmatrix}
\begin{pmatrix}
0 & \bv A & \bv C^T \\
\bv A^T & 0 & \bv B \\
\bv C & \bv B^T & 0 \\
\end{pmatrix}
\\ &=
\begin{pmatrix}
\bv C^T \bv B^T \bv A^T + \bv A\bv B\bv C & \_ & \_ \\
\_ & \bv B\bv C\bv A + \bv A^T \bv C^T \bv B^T & \_ \\
\_ & \_ & \bv B^T \bv A^T \bv C^T + \bv C\bv A\bv B \\
\end{pmatrix},
\end{align*}
where in the last line, we do not calculate the off-diagonal blocks as they do not impact the trace. 
From the cyclic property of trace, we have that $\tr(\bv A\bv B\bv C) = \tr(\bv B\bv C\bv A) = \tr(\bv C\bv A\bv B)$. Further, since the trace is linear and since taking the transpose of a block does not change its trace, we have $\tr(\bv M^3) = 6\cdot \tr(\bv A\bv B\bv C)$. Thus, since we can compute $\tr(\bv M^3)$ in $T(3n)$ time, we can compute $\tr(\bv A \bv B \bv C)$ in $T(3n) + \tilde O(n^2)$ time as desired. 
\end{proof}

We also state the following basic reduction, which allows us to reduce matrix multiplication over bounded integer matrices to multiplication over a sufficiently large finite field $\F_q$. Note that a reduction in the opposite direction (reducing finite field matrix multiplication to integer matrix multiplication) follows from the properties of modular arithmetic: one can simply take the output modulo $q$ at the end of the computation.

\begin{proposition}[Integer Multiplication via Finite Field Multiplication]\label{prop:convertFFtoInt}
If there exists an algorithm that given any $\bv A, \bv B \in \mathbb{F}_q^{n \times n}$ for $q = \poly(n)$, computes $\bv A\bv B$ in time $O(T(n))$ time, then there exists an algorithm that,  given $\bv A, \bv B \in \mathbb{Z}^{n \times n}$, with entries bounded in magnitude by $\poly(n)$, computes $\bv A\bv B$ in time $O(T(n)) + \tilde O(n^2)$.
\end{proposition}

\begin{proof}

Assume that $\bv A$ and $\bv B$ are nonnegative with entries bounded in magnitude by $w = \poly(n)$. Fix $q = 2^k$ such that $w^2 n \le q \le 2 w^2 n$ so that $\F_q$ is a finite field. One can check that the entries of $\bv A \bv B$ are bounded by $w^2 n$. Thus, $\bv A \bv B \mod q = \bv A \bv B$. So, if we simply let $\bv A_q,\bv B_q \in \F_q^{n \times n}$ represent $\bv A$ and $\bv B$ with their entries interpreted as elements in $\F_q$, we can compute $\bv A_q \bv B_q = \bv A \bv B \mod q = \bv A \bv B$ in $O(T(n))$ time, since $q = O(w^2n) = \poly(n)$. This gives the result in this case.

To handle $\bv A, \bv B \in \Z^{n \times n}$ with negative entries, we can simply shift them. Let $\bv 1$ be the $n \times n$ all ones matrix. Let $\bv{A'} = \bv A + w\bv1$ and $\bv {B'} = \bv {B} + w \bv 1$. Since $\bv A'$ and $\bv B'$ are nonnegative with entries bounded in magnitude by $2w = \poly(n)$, we can multiply them in time $O(T(n))$ using the reduction above. We can then extract the matrix product $\bv A \bv B$ using the formula:
\begin{align*}
\bv A'\bv B' &= \bv A \bv B + w(\bv A \bv 1) + w(\bv B \bv 1) + w^2 \bv 1^2\\
\bv A \bv B &= \bv A' \bv B' - w(\bv A \bv 1) - w(\bv B \bv 1) - w^2 \bv 1^2.
\end{align*}

We can compute $w (\bv A \bv 1)$, $(w \bv B \bv 1)$ and $w^2 \bv 1^2 $ all in $\tilde O(n^2)$ time since all columns of $\bv 1$ are identical to each other. Thus, from the matrix product $\bv A'\bv B'$ we can extract $\bv A \bv B$ in $\tilde O(n^2)$ time. The total runtime is $O(T(n)) + \tilde O(n^2)$, giving the proposition
\end{proof}

Finally, we note that we can reduce computing $\tr(\bv A^3)$ for symmetric $\bv A$ (and in turn $\tr(\bv A \bv B \bv C)$ by \Cref{prop:trace}) to many scalar output linear algebra problems, using the results of \cite{musco2017spectrum}.

\begin{theorem}[Adapted from \cite{musco2017spectrum}]\label{prop:musco}
If there exists an algorithm that, given symmetric $\bv A \in \Z^{n \times n}$ with entries bounded in magnitude by $\poly(n)$ computes a $\frac{1}{\poly(n)}$ relative error approximation to any of the following functions in time $O(T(n))$, then there exists an algorithm that, given $\bv A, \bv B \in \Z^{n \times n}$ with entries bounded in magnitude by $\poly(n)$, computes $\bv A \bv B$ in time $O(T(n)) + \tilde O(n^2)$.
\begin{itemize}
    \item $\det(\bv A)$, $\log\det(\bv A)$
    \item $\tr(\bv A^{-1})$ for non-singular $\bv A$, $\tr(\exp(\bv A))$, $\tr(\bv A^p)$ for integer $p \ge 3$
    \item The $p^{th}$ Schatten norm, $\|\bv A\|_p^p = \sum_{i=1}^n \sigma_i(\bv A)^p$ for any fixed $p \neq 1,2$.
    \item The SVD entropy $\sum_{i=1}^n \sigma_i(\bv A) \cdot \log\sigma_i(\bv A))$ for non-singular $\bv A$.
    \end{itemize}
    Above $\sigma_1(\bv A),\ldots, \sigma_n(\bv A) \ge 0$ are the singular values of $\bv A$.
\end{theorem}
\begin{proof}
    We apply Theorem 15 (and the resulting corollaries) of \cite{musco2017spectrum} with a few minor modifications. They state their reduction from triangle counting (i.e., determining if $\tr(\bv A^3) > 0$ for symmetric binary $\bv A$) rather than computing $\tr(\bv A^3)$. However, it is not hard to check that the reduction indeed exactly computes $\tr(\bv A^3)$ and holds as long as $\bv A$ is symmetric and has integer entries bounded by $\poly(n)$. Additionally, their reductions determine $\tr(\bv A^3)$ by computing $f(\bv B)$  for $\bv B = \bv I - \delta \bv A$ where $\delta = \frac{1}{\poly(n)}$. In our setting, we need $\bv B$ to have integer entries. So, we instead work with $\bv B = \frac{1}{\delta} \bv I - \bv A$, which is integer, with $\poly(n)$ bounded entries. Finally, \cite{musco2017spectrum} uses the real-RAM model, and thus does not consider the cost of representing the entries in $f(\bv B)$. As discussed, $f(\bv B) = \det(\bv B)$ and $f(\bv B) = \tr(\bv B^{-1})$ can  potentially be as large as $\exp(\poly(n))$. Nevertheless, using a floating point representation, we can represent $f(\bv B)$ to $\frac{1}{\poly(n)}$ relative accuracy using $O(\log n)$ bits, and perform arithmetic operations on its value in $\tilde O(\log n)$ time.
\end{proof}

\subsection{Proof of Main Results}\label{sec:mainProofs}

With the above preliminaries in place, we now prove our main results, \Cref{thm:intro} and \Cref{cor:intro}.

\begin{proof}[Proof of \Cref{thm:intro}]
Given $\bv A, \bv B \in \Z^{n \times n}$ with $\poly(n)$ bounded entries, we apply \Cref{prop:convertFFtoInt} to reduce computing their product to computing the product of $\bv A_q,\bv B_q \in \F_q^{n \times n}$ for some $q = \poly(n)$. We in turn reduce computing this product to computing $\tr(\bv A_q \bv B_q \bv C_q)$ for $\bv A_q, \bv B_q, \bv C_q \in \F_q^{n \times n}$ via Bernstein-Vazirani, i.e., \Cref{theorem:QuantTrace}. Finally, we reduce $\tr(\bv A_q \bv B_q \bv C_q)$ to computing $\tr(\bv M_q^3)$ for symmetric $\bv M_q \in \F_q^{3n \times 3n}$ by \Cref{prop:trace}. This operation trivially reduces to computing $\tr(\bv {M}^3)$ for symmetric $\bv {M} \in \Z^{n \times n}$ with $\poly(n)$ bounded entries by taking the result modulo $q$. Each reduction incurs just a constant factor multiplicative runtime increase plus an additive $\tilde O(n^2)$, giving our final claimed runtime of $O(T(n)) + \tilde O(n^2)$.
\end{proof}

\begin{proof}[Proof of \Cref{cor:intro}]
   We reduce computing  $\bv A, \bv B \in \Z^{n \times n}$ with $\poly(n)$ bounded entries to computing $\tr(\bv M^3)$ for symmetric $\bv M$ with $\poly(n)$ bounded entries via \Cref{thm:intro}. We in turn reduce this problem to any of the listed problems via \Cref{prop:musco}, giving the result.
\end{proof}

\section{Conclusion and Open Questions}

Our observations help improve our understanding of why the complexity of solving many scalar-output linear algebra problems to high precision seems to be limited by $O(n^\omega)$. There are many interesting open questions concerning the fine-grained complexity of these linear algebraic problems, on both classical and quantum computers, including:

\begin{itemize}
\item Can we give a quantum reduction from matrix multiplication to computing $\bv A^{-1} \bv x$ (i.e., solving a linear system)? One approach is to use Bernstein-Vazirani to extract $\bv A^{-1}$ from the product $\bv A^{-1} \bv x$, and then use the classic reduction from matrix-multiplication to inversion. A similar idea with matrix-vector products has been used e.g., in \cite{childs2021quantum,lee2021quantum}. However, the challenge becomes that any algorithm for computing $\bv A^{-1} \bv x$ will inherently be approximate, and the robustness of Bernstein-Vazirani to approximation error is unclear. For related work giving a classical reduction from the problem of rank finding (for which our best algorithms run in $O(n^\omega)$ time) to linear system solving, see \cite{bafna2021optimal}. It is not known how to reduce matrix multiplication to rank finding on a classic computer or an algebraic circuit -- can we find such a reduction on quantum computers?
\item Related to the above, a noise robust version of Bernstein-Vazirani would conceivably allow computing the gradient of any function using just two black-box calls to the function via finite differencing. This would allow us to translate the Baur-Strassen theorem for algebraic circuits (and any reductions that it implies) to quantum computers. For related work on the connection between Bernstein-Vazirani and gradient computation, see \cite{jordan2005fast}.
\item Can we give any classical or quantum reduction from matrix multiplication to computing the spectral norm $\norm{\bv A}_2$ (i.e., the largest singular value of $\bv A$) to $\frac{1}{\poly(n)}$ relative error? This quantity is notably not covered by the reductions of \cite{musco2017spectrum}.
\item Finally, an appealing but potentially difficult open question is to improve existing classical reductions. We know that an $O(n^{3-\delta})$ time algorithm for $\tr(\bv A^3)$ gives an $O(n^{3-\delta/3})$ time algorithm for Boolean matrix multiplication. Can we tighten this?
\end{itemize}

\subsubsection*{Acknowledgements} Both authors were partially supported by NSF Grant 2046235.




\bibliography{bib}
\bibliographystyle{alpha}

\appendix

\section{Bernstein-Vazirani Over Finite Fields}\label{app:quantum}

For completeness, we include a proof of the Bernstein-Vazirani theorem \cite{bernstein1993quantum} for linear functions over $\F_q$. We start with some basic quantum computing preliminaries.

\smallskip

\noindent \textbf{Quantum Representations of Field Elements.} We will represent $k \in \{0,\ldots,q-1\}$ as an element in $\F_q$ using $\lceil\log_2(q) \rceil$ qubits. The corresponding quantum state will be written as $\ket{k}= \ket{k_1,k_2,...k_{\lceil \log_2(q) \rceil}} = \ket{k_1} \otimes \ket{k_2} \otimes ...\otimes ... \ket{ k_{\lceil \log_2(q) \rceil }}$, where each $k_i$ is a qubit representing a binary value and $\otimes$ is the tensor product. For $\bv y \in \F_q^n$, its quantum state representation will be written as $\ket{\bv y} = \ket{\bv y_1,\bv y_2,\ldots,\bv y_n} = \ket{\bv y_1} \otimes \ket{\bv y_2} \otimes \ldots \otimes \ket{\bv y_n}$, where $\bv y_1,\ldots,\bv y_n \in \F_q$ are the entries of $\bv y$.




%

\begin{definition}[Quantum Fourier Transform over $\F_q$]\label{def:qft}
Let $\omega_q = e^{\frac{2\pi i}{q}}$. The quantum Fourier transform $QFT_q$ maps any state $\ket{x} = \sum_{k=0}^{q-1} x_k \ket k$ to $QFT_q \ket{x} = \sum_{j=0}^{q-1} y_j \ket j$, with
\begin{equation}
    y_j = \dfrac{1}{\sqrt{q}} \sum_{k=0}^{q-1} x_k \omega^{jk}_q.
\end{equation}
The inverse quantum Fourier transform $QFT_q^{-1}$ similarly maps $\ket{x}$ to $QFT_q^{-1} \ket{x} = \sum_{j=0}^{q-1} y_j \ket{j}$ with $y_j = \frac{1}{\sqrt q} \sum_{k=0}^{q-1} x_k \omega_q^{-jk}$. 
The cost of computing $QFT_q$ or $QFT_q^{-1}$ with a quantum circuit is $O(\log q \cdot \log(\log q)^2 \cdot \log\log \log q)$ \cite{cleve2000fast}.
\end{definition}
The QFT is extended to length-$n$ vectors via tensoring. In particular, for $\ket{\bv x} \in \F_q^n$, $QFT_q^{\otimes n}\ket{\bv x} = \bigotimes_{i=1}^n QFT_q \ket{\bv x_i}$. The runtime of computing this operation is just the runtime of computing each of the $n$ individual QFTs, so $\tilde O(n \log q)$ time. The extension of $QFT_q^{-1}$ to vectors is analogous.

\begin{definition}[Phase Oracle]\label{def:phase}
    A phase oracle for a function $f: \F_q^n \rightarrow F_q$ is a unitary transformation $U_f$ defined by $U_f \ket{\bv x} = \omega_q^{f(\bv x)} \ket{\bv x}$ for $\omega_q = e^{\frac{2\pi i}{q}}$.
\end{definition}

We can now state the Berstein-Vazirani theorem over general finite fields. \Cref{thm:ExtendedBV} is a simplified version of this full statement.

\begin{theorem}[Bernstein-Vazirani Theorem]\label{thm:ExtendedBVFull}
For any $f:\mathbb{F}_q^n\rightarrow \mathbb{F}_q$ of the form $f(\bv x)=\langle \bv s, \bv x \rangle$, where $\bv s\in\mathbb{F}_q^n$ is a fixed vector, there exists a quantum algorithm that outputs $\bv s$ with probability $1$ using just a single call to an phase oracle (\Cref{def:phase}) for $f$ and $\tilde O(n \log q)$ additional time.
\end{theorem}

\begin{proof}
The algorithm start with the all zeros vector $\ket{0}^{\otimes n}$. After applying the quantum Fourier transform to each entry, we obtain an equal superposition over all vectors in $\F_q^n$:
\begin{equation*}
    QFT_q^{\otimes n}\ket{\bv x} = QFT_q^{\otimes n}\ket{0}^{\otimes n} = \left [ \dfrac{1}{\sqrt{q}} \sum_{y=0}^{q-1}\ket{y}\right ]^{\otimes n} = \dfrac{1}{\sqrt{q^n}} \sum_{\bv z \in \mathbb{F}_q^n}\ket{\bv z}.
\end{equation*}
We then apply our phase oracle for $f(\bv x) = \langle \bv s, \bv x \rangle$ to obtain:
\begin{equation}\label{eq:phase}
    U_f \; QFT_q^{\otimes n} \ket{0}^{\otimes n} = \dfrac{1}{\sqrt{q^n}}\sum_{\bv z \in \mathbb{F}_q^n} \omega_q^{f(\bv z)}\ket{\bv z} = \dfrac{1}{\sqrt{q^n}}\sum_{\bv z \in \mathbb{F}_q^n} \omega_q^{\langle \bv s, \bv  z\rangle}\ket{\bv z}.
\end{equation}
Since $\omega_q^{\langle \bv s,\bv z\rangle} = \prod_{i=1}^n \omega_q^{\bv s_i \bv z_i}$ we can expand \eqref{eq:phase} as a tensor product to obtain:
\begin{align}\label{eq:expansion}
U_f \; QFT_q^{\otimes n} \ket{0}^{\otimes n} &= \sum_{\bv z \in \F_q^n} \bigotimes_{i=1}^n \left ( \frac{1}{\sqrt{q}} \omega_q^{\bv s_i \bv z_i} \ket{\bv z_i} \right )\nonumber \\ &= \bigotimes_{i=1}^n \left (\sum_{z =0}^{q-1} \frac{1}{\sqrt{q}} \omega_q^{\bv s_i z} \ket{ z} \right )\nonumber \\ &= \bigotimes_{i=1}^n QFT_q \ket{\bv s_i} = QFT_q^{\otimes n} \ket{\bv s}.
\end{align}
We now simply apply the inverse quantum Fourier transform to extract the vector $\bv s$. I.e., plugging into \eqref{eq:expansion}, we have $QFT_q^{-1 \otimes n} \; U_f \; QFT_q^{\otimes n} \ket{0}^{\otimes n} = QFT_q^{-1 \otimes n} \; QFT_q^{\otimes n} \ket{\bv s} = \ket{\bv s}.$

The algorithm for computing $\bv s$ thus requires one call to the phase oracle $U_f$ along with one application each of $QFT_q^{\otimes n}$ and $QFT_{q}^{-1 \otimes n}$, taking time $\tilde O(n \log q)$.
\end{proof}

\end{document}